\documentclass[a4paper,11pt]{amsart}
\usepackage{amssymb}
\usepackage{hyperxmp}
\usepackage{graphicx}
\usepackage{amsmath}
\usepackage[pdfdisplaydoctitle=true,
            colorlinks=true,
            urlcolor=blue,
            citecolor=blue,
            linkcolor=blue,
            pdfstartview=FitH,
            pdfpagemode= UseNone,
            bookmarksnumbered=true]{hyperref}
\usepackage{todonotes}
\newtheorem{theorem}{Theorem}[section]
\newtheorem{corollary}[theorem]{Corollary}
\newtheorem{lemma}[theorem]{Lemma}

\theoremstyle{definition}

\theoremstyle{remark}
\newtheorem{remark}{Remark}

\numberwithin{equation}{section}
\newcommand{\ZZ}{\mathbb{Z}} 
\newcommand{\RR}{\mathbb{R}} 

\newcommand{\Circle}{\mathbb{S}^{1}} 
\newcommand{\g}{\mathfrak{g}} 

\newcommand{\id}{\mathrm{id}}
\newcommand{\oed}{\mathcal{O}(\epsilon^{2},\epsilon\delta^{2})}

\newcommand{\DiffS}{\mathrm{Diff}^{\infty}(\mathbb{S}^{1})} 
\newcommand{\CS}{\mathrm{C}^{\infty}(\mathbb{S}^{1})} 
\newcommand{\GS}{C^{\infty}G} 
\newcommand{\gs}{C^{\infty}\mathfrak{g}} 
\newcommand{\D}[1]{\mathcal{D}^{#1}(\mathbb{S}^{1})}
\newcommand{\G}[1]{H^{#1}G}
\newcommand{\HH}[1]{H^{#1}(\mathbb{S}^{1})}

\newcommand{\norm}[1]{\left\Vert#1\right\Vert}
\newcommand{\abs}[1]{\left\vert#1\right\vert}

\DeclareMathOperator{\ad}{ad} %
\hypersetup{
    pdftitle={Two-component equations modelling water waves with constant vorticity}, 
    pdfauthor={J. Escher, D. Henry, B. Kolev and T. Lyons}, 
    pdfsubject={MSC 2010: 35Q35; 76B15;   35Q53; 58D05}, 
    pdfkeywords={Water waves; vorticity; model equations; Euler equation; diffeomorphism group}, 
    pdflang=EN 
    }
\begin{document}

\title[Two-component equations modelling vorticity]{Two-component equations modelling water waves with constant vorticity}

\author{Joachim Escher}
\address{Institute for Applied Mathematics, University of Hanover, D-30167 Hanover, Germany}
\email{escher@ifam.uni-hannover.de}

\author{David Henry}
\address{Department of Applied Mathematics, University College Cork, Western Road, Cork, Ireland}
\email{d.henry@ucc.ie}

\author{Boris Kolev}
\address{Aix Marseille Universit\'{e}, CNRS, Centrale Marseille, I2M, UMR 7373, 13453 Marseille, France}
\email{boris.kolev@math.cnrs.fr}

\author{Tony Lyons}
\address{School of Mathematical Sciences, Dublin Institute of Technology, Kevin Street, Dublin 8, Ireland}
\email{tony.lyons@mydit.ie}



\date{\today}%

\begin{abstract}
In this paper we derive a two-component system of nonlinear equations which models two-dimensional shallow water waves with constant vorticity.  Then we prove well-posedness of this equation using a geometrical framework which allows us to recast this equation as a geodesic flow on an infinite dimensional manifold. Finally, we provide a criteria for global existence.
\end{abstract}
\maketitle

\tableofcontents

\section{Introduction}
\label{sec:introduction}

The focus of this paper is the following nonlinear two-component system of equations which model two-dimensional shallow water waves with constant vorticity:
\begin{equation}\label{eq:Main}
  \left\{
    \begin{split}
      m_{t} & = \alpha u_{x} - au_{x}m - um_{x} - \kappa\rho\rho_{x}, \\
      \rho_{t} & = - u\rho_{x} - (a-1)u_{x}\rho,
    \end{split}
  \right.
\end{equation}
where $m=u-u_{xx}$. Here $a\neq 1$ is a real parameter, $\alpha$ is a constant which represents the vorticity of the underlying flow, and $\kappa>0$ is an arbitrary real parameter. In Sections \ref{sec:derivation} and \ref{sec:system} we present a derivation of system \eqref{eq:Main} in the hydrodynamical setting using formal asymptotic expansions and perturbation theory \cite{Joh2005} applied to the full governing equations for two-dimensional water waves with constant vorticity. The above system generalises and incorporates a number of celebrated nonlinear partial differential equations which have recently been derived as approximate models in hydrodynamics.

When $\alpha=0$ (which in our considerations corresponds to irrotational fluid flow) and $\rho\equiv 0$ we obtain a one-component family of equations which are parameterised by $a\neq 1$. This family of so-called $b-$equations possess a number of structural phenomena which are shared by solutions of the family of equations \cite{EY2008,Hen2008,Hen2009b}. However,  there are just two members of this family which are integrable \cite{Iva2007}: the Camassa-Holm  (CH) \cite{CHH1994,CH1993} equation, when $a=2$, and the Degasperis-Procesi (DP) \cite{DP1999} equation, when $a=3$. The CH equation is a remarkable equation which is both integrable and which possesses both global solutions and solutions which exhibit wave-breaking in finite time \cite{Con2000,CE1998b,CE1998}, a feature which is also shared by the DP equation \cite{ELY2007,ELY2006}. Both the CH and DP equations can be derived in the hydrodynamic setting of shallow water waves \cite{Con2011,CL2009,Joh2002}, and hence they represent the first examples of integrable equations which possess wave-breaking solutions \cite{Con2011}.

In \cite{OR1996} the authors presented a number of integrable multi-component generalisations of the CH equation, the most popular of which corresponds to $\alpha=0,a=2,\kappa=\pm 1$ in \eqref{eq:Main} \cite{CI2008,CLZ2006,ELY2007a,Iva2006,Hen2009}. It was shown in \cite{CI2008} that this two-component generalisation of CH can be derived from a hydrodynamical setting, in which case $\kappa=1$. The question as to whether these  hydrodynamical model equations could be adapted so as to incorporate an underlying vorticity in the flow is a natural one. Physically, vorticity is vital for incorporating the ubiquitous effects of currents and wave-current interactions in fluid motion \cite{TK1997}, and furthermore the past decade has seen a burgeoning in the mathematical analysis of the full-governing equations for water waves
with vorticity--- cf. \cite{Con2011} for an overview of much of this work. From a model equation viewpoint, following \cite{Joh2003}, Ivanov \cite{Iva2009} derived a number of two-component equations for shallow water waves with underlying constant vorticity, including an integrable generalisation of the two-component CH equation. In Sections  \ref{sec:derivation} and \ref{sec:system} of this paper we extend this work to derive \eqref{eq:Main}, a family of two-component systems parameterised by $a$, where $a=2$ generalises the CH equation, and $a=3$ generalises the DP equation.

In Section~\ref{sec:geometry} we present a geometrical interpretation of~\eqref{eq:Main}, which is shown to correspond (when $a=2$) to the \emph{Arnold--Euler equation} of a right-invariant metric on the infinite dimensional Lie group
\begin{equation*}
  (\DiffS \circledS \CS) \times \RR.
\end{equation*}
This geometrical approach goes back to the pioneering work of V. Arnold~\cite{Arn1966}, published in 1966, who recast the equation of motion of a perfect fluid (with fixed boundary) as the geodesic flow on the volume-preserving diffeomorphisms group of the domain. For the little history, Arnold's paper was written (in French) for the bi-century of the 1765's paper of Euler~\cite{Eul1765a} (also written in French) who recast the equation of motion of a free rigid body as the geodesic flow on the rotation group. As acknowledged by Arnold himself, his paper concentrated on the \emph{geometrical ideas} and not on the difficult analytical technicalities that are inherent when one works with an \emph{infinite dimensional diffeomorphisms group} rather than a \emph{finite dimensional Lie group}. In 1970, Ebin \& Marsden~\cite{EM1970}  reconsidered this geometric approach from the analytical point of view. They proposed to consider the group of smooth diffeomorphisms as an \emph{inverse limit of Hilbert manifolds}. The remarkable observation is that, in this framework, the Euler equation (a PDE) can be recast as an ODE (the geodesic equation) on these Hilbert manifolds. Furthermore, following their approach, if we can prove \emph{local existence and uniqueness of the geodesics} (ODE), then the PDE is \emph{well-posed}. This technique has been used in~\cite{Mis2002,CK2003} for the periodic Camassa--Holm equation. It was extended to (non metric) geodesic flows such as the Degasperis--Procesi equation in~\cite{EK2011} and to right-invariant metrics induced by fractional Sobolev norm (non-local inertia operators) in~\cite{EK2014}. It is used in Section~\ref{sec:well-posedness} to establish the well-posedness of the system \eqref{eq:Main}. Finally, in Section~\ref{sec:global-solutions}, we derive some \textit{a priori} estimates on~\eqref{eq:Main} which lead to a criteria for global existence of solutions.

\section{Derivation of the model equation}
\label{sec:derivation}

In this section we apply formal asymptotic methods to the full-governing equations for two-dimensional water waves with an underlying constant vorticity $\alpha$ to  derive the two-component system \eqref{eq:Main}. We choose Cartesian-coordinates with the $z-$axis pointing vertically upwards and the $x-$axis perpendicular to the crest-lines of the waves, the flow being in the positive $x-$direction. We let $z=0$ denote the location of the flat bed and assume that the mean-depth, or the depth of the undisturbed water domain, is given by $h_{0}$. The velocity field of the two dimensional water flow is given by $(u(t,x,z),w(t,x,z))$ with the water's free surface $z=h_{0}+\eta(t,x)$ --- we denote by $\Omega_{\eta}$ the fluid domain with free boundary. For water waves which are large in scale it is reasonable to make the simplifying assumptions of homogeneity (constant density) and non-viscosity (no internal friction forces). We decompose the pressure
\begin{equation*}
  P(x,z,t)=P_{0}+\rho g(h_{0}-z)+p(x,z,t)
\end{equation*}
into the hydrostatic pressure term (where $\rho$ is the fluid density and $P_{0}$ is the constant atmospheric pressure) and the deviation from hydrostatic pressure $p$. The equations governing the motion of the fluid comprise Euler's equation
\begin{subequations}\label{eq:gov}
\begin{align}
  u_{t}+uu_{x}+wu_z=-\frac{1}{\rho}p_{x} \ &\text{ in } \Omega_{\eta},
  \\
  w_{t}+uw_{x}+ww_z=-\frac{1}{\rho}p_z & \text{ in } \Omega_{\eta},
\end{align}
together with the continuity equation
\begin{equation}
  u_{x}+w_z=0 \ \text{ in } \Omega_{\eta},
\end{equation}
and the dynamic and kinematic boundary conditions
\begin{align}
  w=\eta_{t}+u\eta_{x}, \ p=\eta \rho g \  & \text{ on } z=h_{0}+\eta,
  \\
  w=0 \ & \text{ on } z=0.
\end{align}
\end{subequations}

\subsection{Nondimensionalisation and scaling}
\label{subsec:scaling}

The process of nondimensionalisation enables us to reexpress the governing equations \eqref{eq:gov} in terms of dimensionless variables and functions, and so the resulting equations are purely mathematical. The benefit of this nondimensionalisation procedure is that it naturally introduces dimensionless parameters into the mathematical problem, which measure the relative size and significance of the physical terms  characteristic to the water wave problem we consider. These parameters enable  the formal asymptotic expansion procedures which we engage in when we derive \eqref{eq:Main} below. If $a$ is the typical amplitude of the waves we consider, and $\lambda$ represents a characteristic horizontal length scale (e.g. the wavelength) for the waves, we define the dimensionless parameters $\epsilon = a/h_{0}$, $\delta=h_{0}/\lambda$. Then, performing the change of variables
\begin{gather}\label{eq:cov}
  x\rightarrow \lambda x, \quad z \rightarrow z h_{0}, \quad t \rightarrow \frac{\lambda}{\sqrt{gh_{0}}}t, \quad \eta \rightarrow a \eta,
  \\ \nonumber
  u\rightarrow\epsilon \sqrt{gh_{0}}u, \quad w\rightarrow \epsilon \delta \sqrt{gh_{0}}w, \quad p\rightarrow\epsilon \rho g h_{0}p,
\end{gather}
the new scaled variables are nondimensionalised. In terms of the new variables the system \eqref{eq:gov} assumes the form
\begin{align*}
  u_{t}+\epsilon(uu_{x}+wu_z)&=-p_{x},
  \\
  \delta^{2}(w_{t}+\epsilon(uw_{x}+ww_z))&=-p_z,
  \\
  u_{x}+w_z&=0,
  \\
  w&=\eta_{t}+\epsilon u\eta_{x}, \ p=\eta \rho g \quad \text{ on } z=1+\epsilon \eta,
  \\
  w&=0 \quad \text{ on } z=0.
\end{align*}
We note that $(u,w,p,\eta)\equiv (U(z),0,0,0)$ gives a solution of the above system, which represents laminar flows with a flat free surface and with an arbitrary shear. To incorporate undulating waves in the presence of a shear flow,  we make the transformation $u=(U(z)+\epsilon u)$. If $\epsilon=0$ (flat surface), then  we are simply left with a shear flow as above.
The system which describes waves in the presence of a shear flow is given by
\begin{subequations}\label{eqs:S1}
\begin{align}
  \label{eq:S1a} u_{t}+Uu_{x}+wU'+\epsilon(uu_{x}+wu_z)=-p_{x} & \ \text{ in } \Omega_{\epsilon},
  \\
  \label{eq:S1b} \delta^{2}(w_{t}+Uw_{x}+\epsilon(uw_{x}+ww_z))=-p_z & \ \text{ in } \Omega_{\epsilon},
  \\
  \label{eq:S1c} u_{x}+w_z=0 & \ \text{ in } \Omega_{\epsilon},
  \\
  \label{eq:S1d} w=\eta_{t}+(U+\epsilon u)\eta_{x}, \ p=\eta \rho g \ & \text{ on } z=1+\epsilon \eta,
  \\
  \label{eq:S1e} w=0 \ & \text{ on } z=0.
\end{align}
\end{subequations}
For the simplest nontrivial case of a laminar shear flow, we let $U(z)=\alpha z$, $\alpha$ constant and $0\leq z \leq 1$. We choose $\alpha>0$ whereby the underlying current is in the positive $x-$direction. The Burns condition \cite{Bur1953,FJ1970} for the non-dimensionalised scaled variables becomes
\begin{equation*}
  \int_{0}^{1}\frac{dz}{(U(z)-c)^{2}}=1,
\end{equation*}
and so
\begin{equation}\label{eq:ceq}
  c^{2}- \alpha c-1=0,
\end{equation}
giving us the nondimensionalised speed
\begin{equation*}
  c=\frac 1{2}(\alpha\pm \sqrt{\alpha^{2}+4})
\end{equation*}
of the travelling waves in linear approximation. We mention here that the Burns condition arises as a local bifurcation condition for shallow water waves with constant vorticity, cf. \cite{Con2011}. In physical coordinates the vorticity is given by $\omega=(U-u)_z-w_{x}$. Scaling the vorticity by $\omega=\sqrt{h_{0}/g} \omega$ we get
\begin{equation*}
  \omega =\alpha+\epsilon(u_z-\delta^{2}w_{x}).
\end{equation*}
As we seek a solution with constant vorticity we must have
\begin{equation}\label{eq:vort}
  u_z-\delta^{2}w_{x}=0.
\end{equation}
From \eqref{eq:vort}, \eqref{eq:S1c} and \eqref{eq:S1e} we get
\begin{subequations}\label{eqs:S2}
\begin{align}
  u=u_{0}-\delta^{2}\frac{z^{2}}{2}u_{0,xx}+\mathcal O(\epsilon^{2},\delta^4, \epsilon \delta^{2}),
  \\
  w=-zu_{0,x}+\delta^{2}\frac{z^3}{6}u_{0,xxx}+\mathcal O(\epsilon^{2},\delta^4, \epsilon \delta^{2}).
\end{align}
\end{subequations}
Here $u_{0}(x,t)$ is the leading order approximation for $u$ in an asymptotic expansion, and we see from \eqref{eq:vort} that $u_z=0$ when $\delta\rightarrow 0$ and hence $u_{0}$ is independent of $z$. From \eqref{eq:S1d}, together with \eqref{eqs:S2} we get
\begin{subequations}\label{eqs:S3}
\begin{equation}\label{eq:S3a}
  \eta_{t}+\alpha\eta_{x}+\left[(1+\epsilon \eta)u_{0} +\epsilon \frac{\alpha}{2}\eta^{2} \right]_{x}-\delta^{2}\frac{1}{6}u_{0,xxx}=0,
\end{equation}
ignoring terms of $\mathcal O(\epsilon^{2},\delta^4, \epsilon \delta^{2})$. From \eqref{eq:S1b}, \eqref{eq:S1d}, \eqref{eqs:S2} we have, to the same order
\begin{equation*}
  p=\eta-\delta^{2}\left[\frac{1-z^{2}}{2}u_{0,xt}+\frac{1-z^3}{3}\alpha u_{0,xx}\right].
\end{equation*}
Then from \eqref{eq:S1a} we have
\begin{equation}\label{eq:S3b}
  \left(u_{0}-\delta^{2} \frac 12 u_{0,xx}\right)_{t}+\epsilon u_{0}u_{0,x}+\eta_{x}-\delta^{2} \frac \alpha3 u_{0,xxx}=0.
\end{equation}
\end{subequations}
Letting $\delta,\epsilon \rightarrow 0$ in \eqref{eqs:S3} we get
\begin{subequations}\label{eqs:S4}
\begin{align}
  \label{eq:S4a} \eta_{t}+\alpha\eta_{x}+u_{0,x}=0,
  \\
  \label{eq:S4b} u_{0,t}+\eta_{x}=0,
\end{align}
\end{subequations}
giving us
\begin{equation}\label{eq:etaeq}
  \eta_{tt}+\alpha\eta_{tx}-\eta_{xx}=0.
\end{equation}
Travelling wave solutions $\eta(x-ct)$ of \eqref{eq:etaeq} have a speed $c$ which satisfies \eqref{eq:ceq}. There are two possible speeds $c$, corresponding to a left and right running wave respectively. In particular, from \eqref{eq:S4a} we get
\begin{equation}\label{eq:linear}
  \eta = cu_{0}+\mathcal{O}(\epsilon, \delta^{2}),
\end{equation}
when we choose a specific value of the wavespeed $c$. We note here that while $\epsilon$ and $\delta$ are generally not related, it is known \cite{CL2009} that the small-amplitude long-wave scaling in the absence of vorticity is $\epsilon=O(\delta^2)$, while $\epsilon=O(\delta)$ is the long-wave scaling for waves of moderate amplitude.

\section{The two component system}
\label{sec:system}
We introduce the auxiliary variable $\theta$, defined in terms of $\eta$ and $u_{0}$ as follows
\begin{equation}\label{eq:theta1}
    \theta = 1 + \epsilon k_{1}\eta + \epsilon^{2} k_{2}\eta^{2} + \epsilon\delta^{2} k_{3}u_{0,xx},
\end{equation}
which upon squaring and retaining terms to order $\oed$ gives
\begin{equation}\label{eq:theta2}
    \theta^{2} = 1 + \epsilon(2k_{1})\eta + \epsilon^{2}(k_{1}^{2} + 2k_{2})\eta^{2} +\epsilon\delta^{2}(2k_{3})u_{0,xx} +\mathcal{O}(\epsilon^3,\epsilon^{2}\delta^{2}).
\end{equation}
The coefficients $k_{1}$, $k_{2}$ and $k_{3}$ are as yet undetermined, but as we proceed it will be seen that specific values must be assigned to all three if the resulting system is to be integrable. At a later stage, $\theta$ will be rewritten in terms of the physical variable $\rho$, which satisfies one member of the two component system \eqref{eq:Main} we intend to analyse in the following sections.

\subsection{The $\rho$-component}
\label{subsec:rho-deriv}

Rewriting $\eta$ in terms of $\theta$ using \eqref{eq:theta1} we get
\begin{equation}\label{eq:rho-a}
\begin{split}
  \eta_{t} & = \frac{1}{k_{1} \epsilon}\theta_\xi - \epsilon\frac{k_{2}}{k_{1}}\eta\eta_{t} + \delta^{2}\frac{k_{3}}{k_{1}}u_{0,xxt},
  \\
  \eta_{x} & = \frac{1}{k_{1} \epsilon}\theta_{x} - \epsilon\frac{k_{2}}{k_{1}}(\eta^{2})_{x} + \delta^{2}\frac{k_{3}}{k_{1}}u_{0,xxx},
\end{split}
\end{equation}
and using \eqref{eq:S4a}, \eqref{eq:S4b} and \eqref{eq:linear}, we write to leading order
\begin{equation}\label{eq:rho-b}
\begin{split}
  u_{0,xxt} & = -\eta_{xxx} = -cu_{0,xxx} +\mathcal{O}(\epsilon,\delta^{2}),
  \\
  \eta_{t} & = -\alpha\eta_{x} - u_{0,x} + \mathcal{O}(\epsilon,\delta^{2}).
\end{split}
\end{equation}
Substituting the expressions \eqref{eq:rho-b} into equation \eqref{eq:rho-a} and keeping terms to order $\mathcal{O}(\epsilon,\delta^{2})$ we find
\begin{equation}\label{eq:rho-c}
\begin{split}
  \eta_{t} &= \frac{1}{k_{1} \epsilon}\theta_{t} +  \epsilon\frac{2k_{2}(1+\alpha c)}{k_{1}}\eta u_{0,x} +                       \delta^{2}\frac{k_{3}c}{k_{1}}u_{0,xxx}+\oed,
  \\
  \eta_{x} &= \frac{1}{k_{1} \epsilon}\theta_{x} - \epsilon\frac{2k_{2}c}{k_{1}}\eta u_{0,x} - \delta^{2}\frac{k_{3}}{k_{1}}u_{0,xxx}+\oed.
\end{split}
\end{equation}
Then forming the linear combination $\eta_{t} + \alpha\eta_{x}$ using the expressions in \eqref{eq:rho-c} we obtain
\begin{equation}\label{eq:rho-c1}
  \eta_{t}+\alpha\eta_{x} = \frac{1}{\epsilon k_{1}}(\theta_{t}+\alpha\theta_{x}) + \epsilon\frac{2k_{2}}{k_{1}}\eta u_{0,x} + \delta^{2}\frac{k_{3}(c-\alpha)}{k_{1}}u_{0,xxx} + \oed.
\end{equation}
Replacing the $\eta_{t} + \alpha\eta_{x}$ term in equation \eqref{eq:S3a} with corresponding the expression given by relation \eqref{eq:rho-c1} yields
\begin{multline}\label{eq:rho-d}
  \frac{1}{\epsilon k_{1}}\left(\theta_{t} + \alpha\theta_{x}\right) + \left((1+\epsilon\eta)u_{0} + \epsilon                 \frac{\alpha}{2}\eta^{2}+\epsilon\frac{k_{2}}{k_{1}}\eta u_{0}\right)_{x}  =
  \\
  -\delta^{2}\left(\frac{k_{3}(c-\alpha)}{k_{1}}-\frac{1}{6}\right)u_{0,xxx} + \oed=0.
\end{multline}
We impose the constraint
\begin{equation}\label{eq:constraint}
    \frac{k_{3}}{k_{1}} = \frac{1}{6(c-\alpha)},
\end{equation}
thereby eliminating terms of order $\delta^{2}$ which would otherwise introduce dispersion to this particular member of the two component system.
Using the approximation \eqref{eq:linear} and neglecting terms of order $\oed$, we may rewrite equation \eqref{eq:rho-d} as
\begin{equation}\label{eq:rho-e}
    \frac{1}{\epsilon k_{1}}\left(\theta_{t} + \alpha\theta_{x}\right) + \left((1+\epsilon\left(1 + \frac{\alpha c}{2}+\frac{k_{2}}{k_{1}}\right)\eta)u_{0}\right)_{x} + \oed =0.
\end{equation}
Having imposed only one constraint on the coefficients $k_{1}$ and $k_{3}$ so far,  we may also choose $k_{1}$ and $k_{2}$ such that
\begin{equation}\label{eq:rho-f}
  k_{1} = 1 + \frac{\alpha c}{2}+\frac{k_{2}}{k_{1}},
\end{equation}
in which case \eqref{eq:rho-e} becomes
\begin{equation}\label{eq:rho-e1}
  \frac{1}{\epsilon k_{1}}\left(\theta_{t} + \alpha\theta_{x}\right) + \left((1+\epsilon k_{1}\eta)u_{0}\right)_{x} + \oed =0.
\end{equation}
Keeping terms of order $\mathcal{O}(\epsilon,\delta^{2})$ we approximate the factor $1+\epsilon k_{1}\eta$ in \eqref{eq:rho-e1} by $\theta,$ resulting in
\begin{equation}\label{eq:rho-g}
    \frac{1}{\epsilon k_{1}}\left(\theta_{t} + \alpha\theta_{x}\right) + \left(\theta u_{0}\right)_{x}=0.
\end{equation}
Under the Galilean transformation $x\to x-\alpha t,\ t \to t$, equation \eqref{eq:rho-g} becomes
\begin{equation}\label{eq:rho-h}
  \frac{1}{\epsilon k_{1}}\theta_{t} + (\theta u_{0})_{x} = 0.
\end{equation}
We now redefine the auxiliary function $\theta$ in terms of the physical variable $\rho$ as follows
\begin{equation}\label{eq:rho-i}
  \rho^{\frac{1}{a-1}} = \theta, \qquad a \neq 1.
\end{equation}
We substitute \eqref{eq:rho-i} into \eqref{eq:rho-h}, and upon multiplying the resulting expression by $a-1$ we find
\begin{equation}\label{eq:rho-j}
  \frac{1}{\epsilon k_{1}}\rho_{t} + u_{0}\rho_{x} = (1-a)\rho u_{0,x},
\end{equation}
which up to a rescaling is the first member \eqref{eq:2compa} of our two component system \eqref{eqs:2compsys}.

\subsection{The $m$-component}
\label{subsec:m-deriv}

Using \eqref{eq:theta2}, we express $\eta_{x}$ in terms of the auxiliary function $\theta$ to find
\begin{equation}\label{eq:m1a}
    \eta_{x}=\frac{1}{2k_{1}\epsilon}(\theta^{2})_{x} - \epsilon\left(\frac{2k_{2}+k_{1}^{2}}{2k_{1}}\right)(\eta^{2})_{x}-\delta^{2}\frac{k_{3}}{k_{1}}u_{0,xxx} + \mathcal{O}(\epsilon^{2},\epsilon\delta^{2}),
\end{equation}
which when substituted into \eqref{eq:S3b} gives
\begin{multline*}
  \left(u_{0}-\delta^{2}\beta_{0}^{2}u_{0xx}\right)_{t}  + \mathcal{O}(\epsilon^{2},\epsilon\delta^{2}) = -\epsilon \left[\frac{u_{0}^{2}}{2} - \frac{k_{1}^{2}+2k_{2}}{2k_{1}}(\eta^{2})\right]_{x}
  \\
  - \delta^{2}\left[\left(\frac{k_{3}}{k_{1}}+\frac{\alpha}{3}\right)u_{0,x} - k_{0}u_{0,t}\right]_{xx} - \frac{1}{2\epsilon k_{1}}(\theta^{2})_{x}.
\end{multline*}
Here $\beta_{0}^{2} = 1/2 + k_{0}$ and $k_{0}$ is a constant yet to be determined. Using relations \eqref{eq:linear} and \eqref{eq:rho-b} above,  and neglecting terms of order $\mathcal{O}(\epsilon,\delta^{2})$, yields
\begin{multline}\label{eq:m1c}
  \left(u_{0}-\delta^{2}\left(\frac{1}{2}+k_{0}\right)u_{0xx}\right)_{t} + \oed =
  \\
  -\epsilon \left(1-c^{2}\frac{2k_{2}+k_{1}^{2}}{k_{1}}\right) \left(\frac{u_{0}^{2}}{2}\right)_{x}
  \\
  - \delta^{2}\left(\frac{k_{3}}{k_{1}}+\frac{\alpha}{3} + k_{0} c\right)u_{0,xxx} - \frac{1}{2\epsilon k_{1}}(\theta^{2})_{x}.
\end{multline}
At this point we reexpress  the $\theta^{2}$ term in \eqref{eq:m1c} in terms of $\rho^{2}$ using
\begin{equation}\label{eq:m1d}
  \theta = 1 + \epsilon\theta_{0},
\end{equation}
where $\theta_{0}$ is of order $\mathcal{O}(1,\epsilon,\delta^{2})$. Retaining terms to order $\mathcal{O}(\epsilon^{2},\epsilon\delta^{2})$, we may use a Taylor expansion to write
\begin{equation}\label{eq:m1e}
  \theta^{2(a-1)} = 1 + 2(a-1)\epsilon\theta_{0} + (a-1)(2a-3)\epsilon^{2}\theta_{0}^{2} + \mathcal{O}(\epsilon^3,\epsilon^{2}\delta^{2}) = \rho^{2}
\end{equation}
and as such
\begin{equation}\label{eq:m1f}
  \epsilon2\theta_{0} = \frac{1}{a-1}\rho^{2} - \epsilon^{2}(2a-3)\theta_{0}^{2} -\frac{1}{a-1} + \mathcal{O}(\epsilon^3,\epsilon^{2}\delta^{2}).
\end{equation}
In addition it follows from \eqref{eq:theta2} that
\begin{equation}\label{eq:m1g}
  \epsilon2\theta_{0} = \theta^{2} - 1 - \epsilon^{2}\theta_{0}^{2}.
\end{equation}
Comparing \eqref{eq:m1f} and \eqref{eq:m1g}, we find
\begin{equation}\label{eq:m1h}
  \theta^{2} = \frac{1}{a-1}\rho^{2} - \epsilon^{2}2(a-2)\theta_{0}^{2} + \frac{a-2}{a-1} + \mathcal{O}(\epsilon^3,\epsilon^{2}\delta^{2}).
\end{equation}
Moreover comparing \eqref{eq:theta1} and \eqref{eq:m1d}, we find
\begin{equation}\label{eq:m1i}
  \theta_{0} = k_{1}\eta + \epsilon k_{2}\eta^{2} + \delta^{2}k_{3}u_{0,xx},
\end{equation}
which we may replace in \eqref{eq:m1h} to obtain
\begin{equation}\label{eq:m1j}
  \theta^{2} = \frac{1}{a-1}\rho^{2} - \epsilon^{2}2c^{2}k_{1}^{2}(a-2)u_{0}^{2} + \frac{a-2}{a-1} + \mathcal{O}(\epsilon^3,\epsilon^{2}\delta^{2}),
\end{equation}
where we have used \eqref{eq:linear} to write $\eta^{2} = c^{2}u_{0}^{2} + \mathcal{O}(\epsilon,\delta^{2})$ in the term of order $\mathcal{O}(\epsilon^{2})$ above.
Equation \eqref{eq:m1j} allows us to replace the auxiliary function $\theta^{2}$ appearing in \eqref{eq:m1c} with the physical variable $\rho^{2}$, in which case we find
\begin{multline}\label{eq:m1k}
  m_{t} + \frac{1}{2\epsilon k_{1}(a-1)}(\rho^{2})_{x} +\oed =
  \\
  -\epsilon\left(1-c^{2}{\left(\frac{k_{1}^{2}+2k_{2}}{k_{1}}+2(a-2)k_{1}\right)}\right)\left(\frac{u_{0}^{2}}{2}\right)_{x}
  \\
  - \delta^{2}\left(\frac{k_{3}}{k_{1}}+\frac{\alpha}{3} + k_{0} c\right)u_{0,xxx},
\end{multline}
where we define,
\begin{equation}\label{eq:m1l}
  m = u_{0} - \delta^{2}\left(\frac{1}{2}+k_{0}\right)u_{0,xx}.
\end{equation}
Adding the term $\alpha m_{0,x}$ to both sides of (\ref{eq:m1k}) we find
\begin{multline}\label{eq:m1m}
  m_{t} + \alpha m_{0,x} + \frac{1}{2\epsilon k_{1}(a-1)}(\rho^{2})_{x} + \oed =
  \\
  \alpha u_{0,x} + \epsilon\left(1-c^{2}{\left(\frac{k_{1}^{2}+2k_{2}}{k_{1}}+(a-2)k_{1}\right)}\right)\left(\frac{u_{0}^{2}}{2}\right)_{x}
  \\
  + \delta^{2}\left(\frac{k_{3}}{k_{1}}-\frac{\alpha}{6} + k_{0}(c-\alpha)\right)u_{0,xxx},
\end{multline}
where on the right hand side we use $m_{0,x}$ as it appears in (\ref{eq:m1l}).
It was previously stated that $k_{0}$ is an undetermined coefficient, and so we may choose its value such that
\begin{equation}\label{eq:m1n}
  \frac{k_{3}}{k_{1}}-\frac{\alpha}{6} + k_{0}(c-\alpha) = 0.
\end{equation}
Combined with the Burns condition imposed on $c$, as determined by \eqref{eq:ceq}, along with the first constraint on $k_{1}$ and $k_{3}$ in \eqref{eq:constraint}, equation \eqref{eq:m1n} gives
\begin{equation*}
    \beta^{2}_{0} = k_{0} + \frac{1}{2} = \frac{\alpha c - \alpha^{2} - 1}{6(c-\alpha)^{2}} + \frac{1}{2} = \frac{1}{3c^{2}(c-\alpha)^{2}},
\end{equation*}
thus fixing $k_{0}$ uniquely in terms of $\alpha$ and $c$. With $k_{0}$ fixed, \eqref{eq:m1m} becomes
\begin{multline}\label{eq:m1o}
  m_{t} + \alpha m_{x} +\frac{1}{2\epsilon k_{1}(a-1)}(\rho^{2})_{x} + \oed =
  \\
  \alpha u_{0,x} + \epsilon\left(1-c^{2}{\left(\frac{k_{1}^{2}+2k_{2}}{k_{1}}+(a-2)k_{1}\right)}\right)\left(\frac{u_{0}^{2}}{2}\right)_{x}.
\end{multline}
So far we have only imposed two constraints on the coefficients $k_{1}$, $k_{2}$ and $k_{3}$ and so we may choose $k_{1}$ and $k_{2}$ such that,
\begin{equation}\label{eq:m1p}
  k_{1}(1+a) =  1- c^{2}\left(\frac{k_{1}^{2}+2k_{2}}{k_{1}}+2(a-2)k_{1}\right),
\end{equation}
which upon substitution into \eqref{eq:m1o} gives,
\begin{multline}\label{eq:m1q}
  m_{t} + \alpha m_{x} - \alpha u_{0,x} + \epsilon k_{1}(1+a) \left(\frac{u_{0}^{2}}{2}\right)_{x}
  \\
  + \frac{1}{2\epsilon k_{1}(a-1)}(\rho^{2})_{x} + \oed = 0.
\end{multline}
Furthermore applying the Galilean transformation introduced in Section \ref{subsec:rho-deriv} to \eqref{eq:m1q}  gives
\begin{multline}\label{eq:m1q2}
    m_{t}  -\alpha u_{0,x} + \epsilon {k_{1}}(1+a)\left(\frac{u_{0}^{2}}{2}\right)_{x}
    \\
    + \frac{1}{2\epsilon k_{1}(a-1)}(\rho^{2})_{x} + \oed = 0.
\end{multline}
We also note that to leading order we have $u_{0} = m + \mathcal{O}(\delta^{2}),$ and so we may write
\begin{equation}\label{eq:m1r}
    \left(\frac{u_{0}^{2}}{2}\right)_{x} = \frac{1}{1+a}(a u_{0,x}m + u_{0}m_{x}) + \mathcal{O}(\delta^{2}),
\end{equation}
and accordingly \eqref{eq:m1q} may be written as
\begin{equation}\label{eq:m1s}
    m_{t} -\alpha u_{0,x} + \epsilon au_{0,x}m + \epsilon u_{0}m_{x} + \frac{1}{2\epsilon k_{1}(a-1)}(\rho^{2})_{x} + \oed = 0.
\end{equation}
This is the second member \eqref{eq:2compb} of our two component system \eqref{eqs:2compsys}, up to  the following rescaling:
\begin{equation}
    x\to\beta_{0}x,\quad t\to \beta_{0}t,\quad u\to \frac{1}{2\epsilon k_{1}}u,\quad \rho \to \sqrt{\frac{\kappa}{2}(a-1)}\rho,
\end{equation}
where $\kappa > 0$ is an arbitrary parameter. With regard to this rescaling, we note that it is important here that $u,m$ be of order $O(\epsilon)$, otherwise following the rescaling $u,m$ would be $O(1/\epsilon)$ and all subsequent considerations would fail. Under this rescaling \eqref{eq:rho-j} and \eqref{eq:m1s} become,
\begin{subequations}\label{eqs:2compsys}
\begin{align}
  \label{eq:2compa} & \rho_{t} + u\rho_{x} + (a-1)u_{x}\rho = 0,
  \\
  \label{eq:2compb} & m_{t} - \alpha u_{0,x} + au_{0,x}m + u_{0}m_{x} + \kappa\rho\rho_{x} = 0,
\end{align}
\end{subequations}
where terms of order $\mathcal{O}(\epsilon,\delta^{2})$and higher have been neglected.  The coefficients $k_{1}$, $k_{2}$ and $k_{3}$ are completely determined by the constraints \eqref{eq:constraint}, \eqref{eq:rho-f} and \eqref{eq:m1p} in terms of $a,$ $c,$ and $\alpha$ as follows:
\begin{align*}
  k_{1} &= \frac{1}{(1+c^{2})(a+1)}+\frac{c^{2}}{a+1}, \\
  k_{2} &= \left(\frac{1}{(a+1)(1+c^{2})}+\frac{c^{2}(1-a)}{2(a+1)}-\frac{1}{2}\right)k_{1}, \\
  k_{3} &= \frac{k_{1}}{6(c-\alpha)}.
\end{align*}
Following minor relabelling, equations \eqref{eqs:2compsys} are transformed to \eqref{eq:Main}, which constitutes the two component system that will be investigated in the remainder of this article.
\begin{remark}
  In the particular case $a=2$ we reduce to the integrable two-component equation which was derived in~\cite{Iva2009}.
\end{remark}
\begin{remark}
 When $a=\frac{c^4+1}{(c^{2}+1)^{2}}$ we have $k_{2}=0$, and as a result of this a number of transformations become linear, in particular \eqref{eq:theta1}. Accordingly, in this setting it is convenient to express the physical variable $\eta$ in terms of the auxiliary variables $\theta$ or $\rho$, which should be useful in practical applications.
\end{remark}

\section{Geometric reformulation as a geodesic flow}
\label{sec:geometry}

In this section, we recast the system~\eqref{eq:Main} as a geodesic flow on the tangent bundle of a suitable infinite-dimensional Lie group. For simplicity, we will focus on the periodic case (i.e on the circle $\Circle$), but there is no obstacle to working on the non-periodic case (i.e on the real line), provided we correctly impose some conditions at infinity. When $\alpha = 0$, $a=2$ and $\kappa = 1$, the system~\eqref{eqs:2compsys} reduces to the \emph{two component Camassa--Holm equation} (2CH)
\begin{equation}\label{eq:2CH}
  \left\{
  \begin{split}
    m_{t} + um_{x} & = -2u_{x}m - \rho\rho_{x}, \\
    \rho_{t} +u\rho_{x}  & = -u_{x}\rho,
  \end{split}
  \right.
\end{equation}
where $m = u-u_{xx}$ and which was first derived in~\cite{OR1996}. In~\cite{EKL2011}, the system~\eqref{eq:2CH} was recast as the geodesic equations for a right-invariant Riemannian metric on a semi-direct product $\DiffS \circledS \CS$, where $\DiffS$ is the group of orientation-preserving, smooth diffeomorphisms of the circle and $\CS$ is the space of real, smooth functions on $\Circle$. The full system \eqref{eq:Main} (with $\alpha=0$) has been studied in the short-wave limit $m=-u_{xx}$ in \cite{Esc2012}. It is shown there that geometrically the two-component system \eqref{eq:Main} corresponds to a geodesic flow with respect to a linear, right-invariant, symmetric connection on a suitable semi-direct product. In addition, if $a=2$ the connection is compatible with a the metric induced by the norm
\begin{equation*}
  \norm{u_{xx}}_{L^{2}(\Circle)} + \norm{\rho}_{L^{2}(\Circle)},
\end{equation*}
where $u\in\CS/\mathbb{R}$ and $\rho\in \CS$. We find here for \eqref{eq:Main} a similar situation: if $a=2$ the geodesic flow is metric (cf. Theorem~\ref{thm:metric-case}) and if $a\ne 2$, the system \eqref{eq:Main} cannot be realized as a \emph{metric} geodesic flow, at least for a large class of inertia operators.

The above described geometric picture was in fact developed in~\cite{EK2011}, where it had been shown that \emph{every quadratic evolution equation} defined on the Lie algebra $\g$ of a Lie group $G$ can be interpreted as the \emph{geodesic equations} for a \emph{right-invariant linear connection} on $G$, although in this case the connection does not derive necessarily from an invariant metric. To be able to use this framework, we rewrite the system~\eqref{eqs:2compsys} as
\begin{equation}\label{eq:EHKL}
  \left\{
    \begin{split}
      m_{t} & = \alpha u_{x} - au_{x}m - um_{x} - \kappa\rho\rho_{x}, \\
      \rho_{t} & = - u\rho_{x} - (a-1)u_{x}\rho, \\
      \alpha_{t} & = 0,
    \end{split}
  \right.
\end{equation}
which is a quadratic evolution equation in the variables
\begin{equation*}
  (u,\rho,\alpha) \in \CS \times \CS \times \RR.
\end{equation*}

We will consider the Fr\'{e}chet Lie group
\begin{equation*}
  \GS:=(\DiffS \circledS \CS) \times \RR,
\end{equation*}
where the group product is given by
\begin{equation*}
  (\varphi_{1},f_{1},s_{1})\ast (\varphi_{2},f_{2},s_{2}):= (\varphi_{1}\circ\varphi_{2},f_{2} + f_{1}\circ\varphi_{2},s_{1}+s_{2})
\end{equation*}
for $(\varphi_{1},f_{1},s_{1}), \,(\varphi_{1},f_{2},s_{2}) \in \DiffS \times \CS \times \RR$, and where $\circ$ denotes composition of mappings. The neutral element is $(\id,0,0)$ and the inverse of an element $(\varphi,f,s)$ is given by
\begin{equation*}
(\varphi,f,s)^{-1}=(\varphi^{-1},-f\circ\varphi^{-1},-s).
\end{equation*}

The Lie algebra of $\GS$ is the vector space $\gs := \CS \oplus \CS \oplus \RR$ endowed with the Lie bracket
\begin{equation*}
[(u_{1},\rho_{1},\alpha_{1}),(u_{2},\rho_{2},\alpha_{2})]:=(u_{1,x}u_{2}-u_{2,x}u_{1},\rho_{1,x}u_{2}-\rho_{2,x}u_{1},0),
\end{equation*}
for $(u_{1},\rho_{1},\alpha_{1}),\,(u_{2},\rho_{2},\alpha_{2})\in\gs$.

Let $R_{g}$ denote the right translation on $\GS$ by the element $g = (\varphi_{2},f_{2},s_{2})$, and $T_{h}R_{g}$ be its tangent map at $h = (\varphi_{1},f_{1},s_{1})$. Then, we have
\begin{equation}\label{eq:tangent-map}
  T_{h}R_{g}.V = (v_{1} \circ\varphi_{2},\sigma_{1} \circ\varphi_{2},\alpha_{1}),
\end{equation}
where $V = (v_{1},\sigma_{1},\alpha_{1})$. Note that its expression does not depend on $h$. In particular, we will have a similar expression for the second order tangent map. More precisely, if $X = (\delta \varphi_{1},\delta f_{1},\delta s_{1},\delta v,\delta \sigma,\delta \alpha)$, we get
\begin{equation}\label{eq:second-tangent-map}
  T_{V}(TR_{g}). X = (\delta \varphi_{1}\circ\varphi_{2},\delta f_{1}\circ\varphi_{2},\delta s_{1},\delta v_{1}\circ\varphi_{2},\delta \sigma_{1}\circ\varphi_{2},\delta \alpha_{1}).
\end{equation}

Given a smooth path $g(t):=(\varphi(t),f(t),s(t)) \in \GS$, we define its \emph{Eulerian velocity}, which lies in the Lie algebra $\gs$, by
\begin{equation*}
    U(t) = TR_{g^{-1}(t)}.g_{t}(t).
\end{equation*}
Let $(u,\rho,\alpha)$ be the three components of $U$. Then, using~\eqref{eq:tangent-map}, we get
\begin{equation*}
  u = \varphi_{t} \circ \varphi^{-1}, \qquad \rho = f_{t} \circ \varphi^{-1}, \qquad \alpha = s_{t}.
\end{equation*}
Introducing the \emph{Lagrangian velocity} $V := (v, \sigma, \alpha)$, where
\begin{equation*}
  v := \varphi_{t} = u \circ \varphi, \qquad \sigma := f_{t} = \rho \circ \varphi, \qquad \alpha := s_{t},
\end{equation*}
we have
\begin{equation*}
  v_{t} = (u_{t} + uu_{x}) \circ \varphi, \qquad \sigma_{t} = (\rho_{t} + u\rho_{x}) \circ \varphi,
\end{equation*}
and we may rewrite equation~\eqref{eq:EHKL} as
\begin{equation*}
  \left\{
    \begin{split}
      v_{t} & = \left(A^{-1}([A,u]u_{x} + \alpha u_{x} - au_{x}Au - \kappa\rho\rho_{x})\right)\circ\varphi, \\
      \sigma_{t} & = \Big((1-a)u_{x}\rho\Big)\circ\varphi, \\
      \alpha_{t} & = 0,
    \end{split}
  \right.
\end{equation*}
where $A = 1 - D^{2}$ and $[A,u]v := A(uv) - u A(v)$. Therefore, the system~\eqref{eq:EHKL} is equivalent to
\begin{equation}\label{eq:geodesic-equations}
    \begin{cases}
        \ \partial_t \, (\varphi,f,s) = (v,\sigma,\alpha), \\
        \ \partial_t \, (v,\sigma,\alpha) = S_{(\varphi,f,s)}(v,\sigma,\alpha),
    \end{cases}
\end{equation}
where
\begin{equation}\label{eq:S-equivariance}
  S_{(\varphi,f,s)} := TTR_{(\varphi,f,s)} \circ S \circ TR_{(\varphi,f,s)^{-1}},
\end{equation}
and
\begin{multline}\label{eq:S-expression}
  S(u,\rho,\alpha) := S_{(\id,0,0)}(u,\rho,\alpha)
  \\
  \Big(A^{-1}([A,u]u_{x} + \alpha u_{x} - au_{x}Au - \kappa\rho\rho_{x}), (1-a)u_{x}\rho, 0\Big).
\end{multline}
The \emph{second order vector field}
\begin{equation}\label{eq:spray}
  F(g,V) = (g,V,V,S_{g}(V))
\end{equation}
defined on
\begin{equation*}
  T\GS \simeq (\DiffS \times \CS \times \RR) \times (\CS \oplus \CS \oplus \RR)
\end{equation*}
is quadratic on $V$ and is called a \emph{spray}. It corresponds to the \emph{geodesic flow of a symmetric linear connection} on $\GS$ (see \cite[Section IV.3]{Lan1999}).

\begin{remark}
Note that the relation~\eqref{eq:S-equivariance} is just the fourth component of the equation
\begin{equation*}
  F(TR_{g}.U) = TTR_{g}.F(U),
\end{equation*}
where $g = (\varphi,f,s)$, $U = (u,\rho,\alpha)$ and which traduces the right-invariance of the spray $F : T\GS \to TT\GS$.
\end{remark}

When the parameter $a=2$, the spray $F$ derives from a right-invariant metric on $\GS$ and the system of equations~\eqref{eq:EHKL} corresponds to the \emph{Arnold-Euler equations} of this metric. Let us briefly recall this formalism. Given a Lie group $G$ and  its Lie algebra $\g$, any inner product on $\g$ induces a \emph{right-invariant} Riemannian metric on $G$ by extending it by right-translation. If this inner product on $\g$ is represented by an invertible operator $\mathbb{A} : \g \to \g^{\ast}$, for historical reasons, going back to the work of Euler on the motion of the rigid body, this operator $\mathbb{A}$ is called the \emph{inertia operator} of the system. Then a path $g(t)$ on $G$ is a geodesic for this right-invariant Riemannian metric, if and only if, its Eulerian velocity $u(t) := TR_{g(t)^{-1}}. \dot{g}(t)$ is solution of the so-called Arnold-Euler equation
\begin{equation*}
  u_{t} = - \mathbb{B}(u,u),
\end{equation*}
where
\begin{equation*}
  \mathbb{B}(u_{1},u_{2}) := \frac{1}{2}\Big( \ad_{u_{1}}^{\top} u_{2} + \ad_{u_{2}}^{\top} u_{1} \Big),
\end{equation*}
where $u_{1},u_{2} \in \g$ and $\ad_{u}^{\top}$ is the adjoint of the operator $\ad_{u}$, with respect to $\mathbb{A}$ (see~\cite{Arn1966} for instance).

\begin{theorem}\label{thm:metric-case}
Given $\kappa>0$, the system~\eqref{eq:EHKL} is, for $a=2$, the Arnold-Euler equation on the Lie algebra $\gs$ of the Fr\'{e}chet Lie group $\GS$ associated to the inner product
\begin{multline}\label{eq:inner-product}
  \langle (u_{1},\rho_{1},\alpha_{1}),(u_{2},\rho_{2},\alpha_{2}) \rangle = \int_{\Circle} u_{1}u_{2}\,dx + \int_{\Circle} u_{1,x}u_{2,x}\,dx
  \\
  + \kappa \int_{\Circle} \rho_{1}\rho_{2}\,dx - \frac{1}{2}\int_{\Circle} \left( u_{1}\alpha_{2} + u_{2}\alpha_{1}\right) dx + \frac{1}{2} \alpha_{1}\alpha_{2},
\end{multline}
where $u_{i},\,\rho_{i}\in\CS$ and $\alpha_{i}\in\mathbb{R}$, for $i=1,2$.
\end{theorem}

\begin{remark}
Note that the norm induced by \eqref{eq:inner-product} on $\gs$
\begin{equation*}
  \norm{(u,\rho,\alpha)}^{2} = \int_{\Circle} u_{x}^{2}\,dx + \int_{\Circle} \left(u - \frac{\alpha}{2}\right)^{2} dx + \frac{\alpha^{2}}{2}+\int_{\Circle} \kappa\rho^{2}dx
\end{equation*}
is equivalent to the Hilbert-norm
\begin{equation*}
  \norm{(u,\rho,\alpha)}^{2} = \norm{u}_{H^{1}}^{2} + \norm{\rho}_{L^{2}}^{2} + \abs{\alpha}^{2}.
\end{equation*}
\end{remark}

\begin{proof}
The inner product~\eqref{eq:inner-product} can be rewritten as
\begin{equation*}
  \langle (u_{1},\rho_{1},\alpha_{1}),(u_{2},\rho_{2},\alpha_{2}) \rangle = \int_{\Circle} (u_{1},\rho_{1},\alpha_{1}) \cdot \mathbb{A}(u_{2},\rho_{2},\alpha_{2}) \,dx
\end{equation*}
where the inertia operator $\mathbb{A}$ is defined by
\begin{equation*}
  \mathbb{A}(u,\rho,\alpha) = \left(Au - \frac{\alpha}{2},\kappa\rho,\frac{1}{2}\left(\alpha - \int_{\Circle}u\right)\right),\quad (u,\alpha,\rho)\in \gs,
\end{equation*}
$Au:=u-u_{xx}$, and $\cdot$ stands for the usual inner product in $\RR^{3}$. Set $U_{i} = (u_{i},\rho_{i}, \alpha_{i})$, for $i = 1,2,3$. We have
\begin{equation*}
  \ad_{U_{1}} U_{2} = (u_{1,x}u_{2}-u_{2,x}u_{1},\rho_{1,x}u_{2}-\rho_{2,x}u_{1},0),
\end{equation*}
and therefore, after some integrations by parts, we get
\begin{equation*}
  \langle \ad_{U_{1}} U_{2}, U_{3} \rangle = \int_{\Circle} (u_{2},\rho_{2},\alpha_{2}) \cdot (f,g,0)\, dx ,
\end{equation*}
where

\begin{equation*}
  f = 2u_{1,x}Au_{3} + u_{1}Au_{3,x} - \alpha_{3}u_{1,x} + \kappa \rho_{1,x}\rho_{3}, \qquad g = \kappa (u_{1}\rho_{3})_{x}.
\end{equation*}
Note now that the equation
\begin{equation*}
  \mathbb{A}(\tilde{u},\tilde{\rho},\tilde{\alpha}) = (f,g,0)
\end{equation*}
has a unique solution which is given by
\begin{equation*}
  \tilde{u} = A^{-1}f + \int_{\Circle} f \,dx, \qquad \tilde{\rho} = \frac{1}{\kappa} g, \qquad \tilde{\alpha} = 2 \int_{\Circle} f \,dx.
\end{equation*}
Thus we have  $\ad_{U_{1}}^{\top} U_{3} = (\tilde{u},\tilde{\rho},\tilde{\alpha})$, where
\begin{multline*}
  \tilde{u} =  A^{-1}\Big( 2u_{1,x}Au_{3} + u_{1}Au_{3,x} - \alpha_{3}u_{1,x} + \kappa \rho_{1,x}\rho_{3} \Big)
  \\
  + \int_{\Circle} (u_{1,x}Au_{3} + \kappa\rho_{1,x}\rho_{3}) \,dx,
\end{multline*}
and
\begin{equation*}
  \tilde{\rho} = (u_{1}\rho_{3})_{x} ,\qquad  \tilde{\alpha} = 2\int_{\Circle} (u_{1,x}Au_{3} + \kappa\rho_{1,x}\rho_{3}) \,dx.
\end{equation*}
We conclude therefore that
\begin{equation*}
  \mathbb{B} (U,U) = \Big( A^{-1}( 2u_{x}Au + uAu_{x} - \alpha u_{x} + \kappa \rho_{x}\rho ), (u\rho)_{x}, 0 \Big),
\end{equation*}
where $U = (u, \rho,\alpha)$ and that the equation
\begin{equation*}
  U_{t} = - \mathbb{B} (U,U)
\end{equation*}
is equivalent to \eqref{eq:EHKL}, when $a=2$.
\end{proof}

\section{Well-posedness of the equation}
\label{sec:well-posedness}

Our strategy will be to study the Cauchy problem for the geodesic equations~\eqref{eq:geodesic-equations}. Following Ebin \& Marsden's approach \cite{EM1970}, if we can prove \emph{local existence and uniqueness of geodesics} of the ODE~\eqref{eq:geodesic-equations} on $TG$, then the PDE~\eqref{eq:EHKL} is \emph{well-posed}. To do so, we need to introduce an approximation of the Fr\'{e}chet--Lie group $\DiffS$ by \emph{Hilbert manifolds}. Let $\HH{q}$ be the completion of $\CS$ for the norm
\begin{equation*}
  \norm{u}_{H^{q}}:= \left( \sum_{k \in \ZZ}(1 + k^{2})^{q}\abs{\hat{u}_{k}}^{2} \right)^{1/2},
\end{equation*}
where $q\in \RR, q \ge 0$. We recall that $\HH{q}$ is a multiplicative algebra for $q > 1/2$ (cf. \cite[Theorem 2.8.3]{Tri1983}). This means that
\begin{equation*}
  \norm{uv}_{\HH{q}} \lesssim \norm{u}_{\HH{q}} \norm{v}_{\HH{q}}, \quad u,v \in \HH{q}.
\end{equation*}

A $C^{1}$ diffeomorphism $\varphi$ of $\Circle$ is of class $H^{q}$ if for any of its lifts to $\RR$, $\tilde{\varphi}$, we have
\begin{equation*}
  \tilde{\varphi} - \mathrm{id} \in \HH{q}.
\end{equation*}
For $q > 3/2$, the set $\D{q}$ of $C^{1}$-diffeomorphisms of the circle which are of class $H^{q}$ has the structure of a \emph{Hilbert manifold}, modelled on $\HH{q}$ (see~\cite{EM1970} or \cite{IKT2013}). Contrary to $\DiffS$, the manifold $\D{q}$ is only a \emph{topological group} and \emph{not a Lie group} (composition and inversion in $\D{q}$ are continuous but \emph{not differentiable}). More precisely, the following regularity properties for $q > 3/2$ are well-known (see~\cite{IKT2013} for instance).
\begin{enumerate}
  \item The mapping
    \begin{equation*}
	  u \mapsto R_{\varphi}(u) := u \circ \varphi, \qquad \HH{q} \to \HH{q}
    \end{equation*}
    is smooth, for any $\varphi \in \D{q}$.
  \item The mapping
    \begin{equation*}
	  (u,\varphi) \mapsto u \circ \varphi, \qquad \HH{q+k} \times \D{q} \to \HH{q}
    \end{equation*}
    is of class $C^{k}$.
  \item The mapping
    \begin{equation*}
	  \varphi \mapsto \varphi^{-1}, \qquad \D{q+k} \to \D{q}
    \end{equation*}
    is of class $C^{k}$.
\end{enumerate}

Moreover the following slightly sharpen result was established in~\cite[Corollary B.6]{EK2014}.

\begin{lemma}\label{lem:composition_regularity}
Let $q > 3/2$. Then, the mappings
\begin{equation*}
  (\varphi, v) \mapsto v \circ \varphi, \qquad \D{q} \times \HH{q} \to \HH{q-1}
\end{equation*}
and
\begin{equation*}
  (\varphi, v) \mapsto v \circ \varphi^{-1}, \qquad \D{q} \times \HH{q} \to \HH{q-1}
\end{equation*}
are $C^{1}$.
\end{lemma}

\begin{remark}
The tangent bundle of the Hilbert manifold $\D{q}$ is trivial because if $\mathfrak{t}: T\Circle \to \Circle \times \RR$ is a smooth trivialisation of the tangent bundle of the circle, then
\begin{equation*}
  \Psi : T\D{q} \to \D{q}\times \HH{q}, \qquad \xi \mapsto \mathfrak{t} \circ \xi
\end{equation*}
defines a smooth vector bundle isomorphism (see~\cite[Page~107]{EM1970}).
\end{remark}

As suggested in~\cite{EKL2011}, we shall use the following choice for the Hilbert approximation of $\GS$:
\begin{equation*}
  \G{q} := \D{q} \times \HH{q-1} \times \RR,
\end{equation*}
which is defined for $q > 3/2$.

\begin{theorem}\label{thm:smoothness-spray}
The spray $F$ defined on $T\DiffS$ by~\eqref{eq:spray} extends smoothly to a spray
\begin{equation*}
  F_{q} : T\G{q} \to TT\G{q},
\end{equation*}
for each $q > 3/2$.
\end{theorem}

\begin{proof}
We first observe that the quadratic operator $S$ defined in~\eqref{eq:S-expression} can be rewritten as
\begin{equation*}
  S(u,\rho,\alpha) = \left(\frac{1}{2}A^{-1}D\left(2\alpha u - \kappa\rho^{2} + (a-3)u_{x}^{2} - au^{2}\right),(1-a)\rho u_{x},0\right),
\end{equation*}
where $D$ stands for $d/dx$. Therefore, the two non-vanishing components $S^{1}$, $S^{2}$ of $S_{(\varphi,f,s)}(v,\sigma,\alpha)$ are
\begin{equation*}
  S^{1} = \frac{1}{2}(A^{-1}D)_{\varphi}\left(2\alpha v - \kappa\sigma^{2} + (a-3)\frac{1}{\varphi_{x}^{2}}v_{x}^{2} - av^{2}\right)
\end{equation*}
and
\begin{equation*}
  S^{2} = (1-a)\frac{1}{\varphi_{x}}\sigma v_{x},
\end{equation*}
where
\begin{equation*}
  (A^{-1}D)_{\varphi} := R_{\varphi} \left(A^{-1}D\right) R_{\varphi^{-1}}.
\end{equation*}
Let $q > 3/2$, $\varphi \in \D{q}$, $v \in \HH{q}$ and $\sigma \in \HH{q-1}$. Then, the expression $S^{2}$ belongs to $\HH{q-1}$ because $\HH{q-1}$ is a multiplicative algebra and the expression $S^{1}$ belongs to $\HH{q}$ because, furthermore,
\begin{equation*}
  R_{\varphi^{-1}} \in \mathcal{L}(\HH{q-1},\HH{q-1}), \qquad R_{\varphi} \in \mathcal{L}(\HH{q},\HH{q})
\end{equation*}
by \cite[Lemma B.2]{EK2014}, and
\begin{equation*}
  A^{-1}D = DA^{-1}\in \mathcal{L}(\HH{q-1},\HH{q}).
\end{equation*}
It remains to be shown that $S^{1}$ and $S^{2}$ depend smoothly on $\varphi, v, \sigma$, which will achieve the proof. First, it is clear that the mapping
\begin{equation*}
  (\varphi, v, \sigma) \mapsto (1-a)\frac{1}{\varphi_{x}}\sigma v_{x}, \quad \D{q} \times \HH{q} \times \HH{q-1} \to \HH{q-1}
\end{equation*}
is smooth because $\HH{q-1}$ is a multiplicative algebra. For the same reason, it is clear that
\begin{equation*}
  (\varphi, v, \sigma) \mapsto 2\alpha v - \kappa\sigma^{2} + (a-3)\frac{1}{\varphi_{x}^{2}}v_{x}^{2} - av^{2},
\end{equation*}
from $\D{q} \times \HH{q} \times \HH{q-1}$ to $\HH{q-1}$ is smooth. Therefore, if we can prove that
\begin{equation*}
  (\varphi,w) \mapsto (A^{-1}D)_{\varphi}w, \qquad \D{q} \times \HH{q-1} \to \HH{q}
\end{equation*}
is smooth, we are done, using the chain rule. Now, observe that $(1-D)$ and $(1+D)$ belong to $\mathrm{Isom}(\HH{q},\HH{q-1})$, for $q \ge 1$, and that
\begin{equation*}
  A^{-1}D = \frac{1}{2}\left( (1-D)^{-1} - (1+D)^{-1}\right),
\end{equation*}
because $A = (1-D)(1+D)$. Therefore
\begin{equation*}
  (A^{-1}D)_{\varphi} = \frac{1}{2}\left( (1-D)_{\varphi}^{-1} - (1+D)_{\varphi}^{-1}\right).
\end{equation*}
But
\begin{equation*}
  D_{\varphi} = \frac{1}{\varphi_{x}}D
\end{equation*}
and hence
\begin{equation*}
  \varphi \mapsto (1\pm D)_{\varphi}
\end{equation*}
from $\D{q}$ to $\mathrm{Isom}(\HH{q},\HH{q-1})$ is smooth. Moreover, since the set $\mathrm{Isom}(\HH{q},\HH{q-1})$ is open in $\mathcal{L}(\HH{q},\HH{q-1})$ and the mapping
\begin{equation*}
    P \mapsto P^{-1}, \quad \mathrm{Isom}(\HH{q},\HH{q-1}) \to \mathrm{Isom}(\HH{q-1},\HH{q})
\end{equation*}
is smooth, we conclude that $(1-D)_{\varphi}^{-1}w$ and $(1+D)_{\varphi}^{-1}w$ both depend smoothly on $(\varphi,w)$.
\end{proof}

Applying the Cauchy--Lipschitz (or Picard--Lindel\"{o}f) theorem in the Hilbert manifold $T\G{q}$, we obtain the following local-existence result for the geodesics.

\begin{corollary}\label{cor:local-existence-Hq}
Let $q > 3/2$. For each initial data $(g_{0},V_{0})\in T\G{q}$, there exists a \emph{unique non-extendable} solution
\begin{equation*}
    (g,V)\in C^{\infty}(J_{q},T\G{q}),
\end{equation*}
of the Cauchy problem~\eqref{eq:geodesic-equations} for the geodesic equations, with $g(0)=g_{0}$ and $V(0)=V_{0}$, defined on some \emph{maximal interval of existence} $J_{q} = J_{q}(g_{0},V_{0})$, which is open and contains $0$. Moreover, the solution depends smoothly on the initial data.
\end{corollary}

Finally, taking account of lemma~\ref{lem:composition_regularity}, we obtain the following.

\begin{corollary}\label{cor:well-posedness-Hq}
Let $q > 5/2$ and $\alpha \in \RR$. For each initial data
\begin{equation*}
  (u_{0},\rho_{0})\in \HH{q}\oplus\HH{q-1},
\end{equation*}
there exists a \emph{unique non-extendable} solution
\begin{equation*}
    (u,\rho) \in C^{0}(J_{q},\HH{q}\oplus\HH{q-1}) \cap C^{1}(J_{q},\HH{q-1}\oplus\HH{q-2}),
\end{equation*}
of the evolution equation~\eqref{eq:EHKL}, with $(u(0),\rho(0)) = (u_{0},\rho_{0})$, defined on some \emph{maximal interval of existence} $J_{q} = J_{q}(u_{0},\rho_{0})$, which is open and contains $0$. Moreover, the solution depends continuously on the initial data.
\end{corollary}

Furthermore, using the right-invariance of the spray, an argument due to Ebin \& Marsden \cite{EM1970} implies that the interval of existence $J_{q}$ in theorem~\ref{cor:local-existence-Hq} is independent of the Sobolev index $q$ and that for any smooth initial condition the solution exists in the smooth category. A proof of the following result follows along the lines of the proof of \cite[Theorem 1.1]{EK2014}. We will only sketch it. Let $\Phi_{q}$ be the flow of the spray $F_{q}$ and $R_{\tau}$ be the (right) action of the rotation group $\Circle$ on $\G{q}$, defined by
\begin{equation*}
    \big(R_{\tau} \cdot (\varphi,f,s) \big)(x): = (\varphi(x+\tau),f(x+\tau),s),
\end{equation*}
where $(\varphi,f,s) \in \G{q}$. Note that if $V = (\varphi,f,s,v,\sigma,\alpha)\in T\D{q+1}$, then
\begin{equation*}
  s \mapsto TR_{s} \cdot V, \quad \Circle \to T\G{q}
\end{equation*}
is a $C^{1}$ map, and that
\begin{equation*}
    \frac{d}{ds} R_{\tau}\cdot V = (\varphi_{x},f_{x},0,v_{x},\sigma_{x},0).
\end{equation*}
Moreover, since the spray $F_{q}$ and its flow $\Phi_{q}$ are invariant under this action, we get
\begin{equation*}
    T_{V_{0}}\Phi_{q}(t).(\varphi_{0,x},f_{0,x},0,v_{0,x},\sigma_{0,x},0) = (\varphi_{x}(t),f_{x}(t),0,v_{x}(t),\sigma_{x}(t),0),
\end{equation*}
from which we deduce that if $V_{0}\in T\G{q+1}$ then so is $V(t)$, and conversely (by reversing time). We therefore also have a well-posedness result for \eqref{eq:EHKL} in the smooth category.

\begin{corollary}\label{cor:well-posedness-smooth}
Let $\alpha \in \RR$ and $(u_{0},\rho_{0})\in \CS\times\CS$ be given. Then, there exists an open interval $J$ with $0\in J$ and a unique non-extendable solution
\begin{equation*}
  (u, \rho) \in C^{\infty} \left(J, \CS \times \CS\right)
\end{equation*}
of \eqref{eq:EHKL} with initial datum $(u_{0},\rho_{0})$.
\end{corollary}

\section{Global solutions}
\label{sec:global-solutions}

Given the fact that the Camassa-Holm equation is embedded in the system \eqref{eq:Main}, classical results, cf. \cite{Con2000,CE1998b,CE1998} make clear that not all solutions to \eqref{eq:Main} can exist globally. Global in time solutions to \eqref{eq:Main} in the case $a=2$, $\alpha=0$, $\kappa=1$ have been constructed in \cite{CI2008}.
In this section, we will not construct global solutions, but we provide a criterion for a strong solution to the full system \eqref{eq:Main} to be defined for all time. Most of the proofs are similar to the ones given in~\cite{EK2014a} and will only be sketched. Throughout this section, we work in the smooth category, and we let
\begin{equation*}
  (g,V) \in C^\infty(J,T\GS)
\end{equation*}
be the unique solution of the Cauchy problem~\eqref{eq:geodesic-equations}, emanating from
\begin{equation*}
  (g_{0},V_{0}) \in T\GS,
\end{equation*}
and defined on the \emph{maximal time interval} $J$. The corresponding solution $(u,\rho)$ of the Arnold-Euler equation~\eqref{eq:EHKL} is a path
\begin{equation*}
    (u,\rho) \in C^{\infty}(J,\CS\oplus\CS),
\end{equation*}
with $(u(0),\rho(0)) = (u_{0},\rho_{0})$. The \emph{momentum} $m(t):=Au(t)$ is defined as a path
\begin{equation*}
  m \in C^{\infty}(J,\CS).
\end{equation*}
We will now establish several \textit{a priori} estimates that will lead to a global existence result.

\begin{lemma}\label{lem:conservation-law}
We have
\begin{equation*}
  (\rho\circ\varphi)(t) \cdot \varphi_{x}^{a-1}(t) = (\rho\circ\varphi)(0) \cdot \varphi_{x}^{a-1}(0),
\end{equation*}
for all $t \in J$.
\end{lemma}

\begin{proof}
Let $\varphi$ be the flow of the time dependent vector field $u$. We have
\begin{equation*}
  \frac{d}{dt}\left(\rho\circ\varphi \cdot \varphi_{x}^{a-1}\right) = \left(\rho_{t} + \rho_{x}u + (a-1)\rho u_{x}\right)\circ \varphi \cdot \varphi_{x}^{a-1},
\end{equation*}
which vanishes because $(u,\rho)$ solves \eqref{eq:EHKL}.
\end{proof}

\begin{corollary}\label{cor:rho-positivity}
Suppose that $\rho_{0} > 0$. Then $\rho(t) > 0$ for all $t \in J$.
\end{corollary}

\begin{proof}
Due to lemma~\ref{lem:conservation-law}, we have
\begin{equation*}
  (\rho\circ\varphi)(t) = (\rho\circ\varphi)(0)\frac{\varphi_{x}^{a-1}(0)}{\varphi_{x}^{a-1}(t)} > 0,
\end{equation*}
which achieves the proof because $\varphi$ is an orientation-preserving diffeomorphism.
\end{proof}

\begin{corollary}\label{cor:rho-bound}
Suppose that $\norm{u_{x}(t)}_{\infty}$ is bounded on any bounded subinterval of $J$. Then $\norm{\rho(t)}_{\infty}$ is bounded on any bounded subinterval of $J$.
\end{corollary}

\begin{proof}
Let $\varphi$ be the flow of the time dependent vector field $u$ and set
\begin{equation*}
  \alpha(t) = \max_{x\in\Circle} 1/ \varphi_{x}(t).
\end{equation*}
Note that $\alpha$ is a continuous function. Let $I$ denote any bounded subinterval of $J$, and set
\begin{equation*}
  K = \sup_{t \in I} \norm{u_{x}(t)}_{\infty}.
\end{equation*}
From equation $\varphi_{t} = u \circ \varphi$, we deduce that
\begin{equation*}
   \left( 1/\varphi_{x} \right)_{t}= -(u_{x}\circ \varphi) / \varphi_{x},
\end{equation*}
and therefore, we have
\begin{equation*}
  \alpha(t)  \le \alpha(0) + K \int_{0}^{t}\alpha(s)\, ds.
\end{equation*}
Thus $\norm{1/\varphi_{x}(t)}_{\infty}$ is bounded on $I$ due to Gronwall's lemma and the conclusion follows from lemma~\ref{lem:conservation-law}.
\end{proof}

\begin{lemma}\label{lem:Hk-apriori-estimate}
Suppose that $\norm{u_{x}(t)}_{\infty}$ is bounded on any bounded subinterval of $J$. Then
\begin{equation*}
  \norm{m(t)}_{H^{k}}^{2} + \norm{\rho(t)}_{H^{k+1}}^{2}
\end{equation*}
is bounded on any bounded subinterval of $J$, for all non-negative integer $k$.
\end{lemma}

\begin{proof}
We have
\begin{multline*}
  \frac{d}{dt}\left(\norm{m}_{L^{2}}^{2} +
\norm{\rho}_{H^{1}}^{2}\right) = \alpha \int u_{x}m - a \int u_{x}m^{2}
- \int umm_{x}
  \\
  - \kappa\int\rho\rho_{x}m - \int u\rho\rho_{x} - (a-1)\int
u_{x}\rho^{2} - \int u_{x}\rho_{x}^{2}
  \\
  - \int u\rho_{x}\rho_{xx} - (a-1)\int u_{xx}\rho\rho_{x} - (a-1)\int
u_{x}\rho_{x}^{2}.
\end{multline*}
After integrating by parts the terms
\begin{equation*}
  - \int umm_{x} = \frac{1}{2}\int u_{x}m^{2}, \quad \text{and} \quad -
\int u\rho_{x}\rho_{xx} = \frac{1}{2} \int u_{x}\rho_{x}^{2},
\end{equation*}
using the Cauchy--Schwarz inequality, and the fact that
$\norm{u}_{H^{2}} \lesssim \norm{m}_{L^{2}}$, we obtain
\begin{equation*}
  \frac{d}{dt}\left(\norm{m}_{L^{2}}^{2} +
\norm{\rho}_{H^{1}}^{2}\right) \lesssim \left(1 + \norm{u_{x}}_{\infty} +
\norm{\rho}_{\infty}\right) \left(\norm{m}_{L^{2}}^{2} +
\norm{\rho}_{H^{1}}^{2}\right)
\end{equation*}
and the conclusion holds from corollary~\ref{cor:rho-bound} and by
Gronwall's lemma. Suppose now that $k \ge 1$. We will show by induction on $k$ that
\begin{multline}\label{eq:Gronwall-estimate}
  \frac{d}{dt}\left(\norm{m}_{H^{k}}^{2} + \norm{\rho}_{H^{k+1}}^{2}\right) \lesssim
  \\
  \left(1 + \norm{m}_{H^{k-1}} + \norm{\rho}_{H^{k}}\right) \left(\norm{m}_{H^{k}}^{2} + \norm{\rho}_{H^{k+1}}^{2}\right),
\end{multline}
which will achieve the proof, using Gronwall's lemma. For $k=1$, we
only need to estimate the term
\begin{equation*}
  \int m_{x}m_{tx} + \int \rho_{xx}\rho_{txx},
\end{equation*}
because
\begin{equation*}
  \norm{u_{x}}_{\infty} + \norm{\rho}_{\infty} \lesssim
\norm{m}_{L^{2}} + \norm{\rho}_{H^{1}}.
\end{equation*}
We have first
\begin{multline*}
  \int m_{x}m_{tx} =  \alpha \int u_{xx}m_{x} - a\int u_{xx}mm_{x} -
a\int u_{x}m_{x}^{2}
  \\
  - \int u_{x}m_{x}^{2} - \int um_{xx}m_{x} - \kappa \int
\rho_{x}^{2}m_{x} - \kappa \int \rho \rho_{xx}m_{x}.
\end{multline*}
Integrating by parts the term
\begin{equation*}
  - \int um_{xx}m_{x} = \frac{1}{2}\int u_{x}m_{x}^{2},
\end{equation*}
using the Cauchy--Schwarz inequality and the fact that $\HH{1}$ is
a multiplicative algebra, we get
\begin{multline*}
  \int m_{x}m_{tx} \lesssim \norm{u}_{H^{2}}\norm{m}_{H^{1}} +
\norm{u}_{H^{2}}\norm{m}_{H^{1}}^{2} +
\norm{u_{x}}_{\infty}\norm{m}_{H^{1}}^{2}
  \\
  + \norm{\rho_{x}^{2}}_{L^{2}}\norm{m}_{H^{1}} +
\norm{\rho}_{\infty}\norm{\rho}_{H^{2}}\norm{m}_{H^{1}}.
\end{multline*}
Now, using the following estimates
\begin{equation*}
  \norm{u}_{H^{2}} \lesssim \norm{m}_{L^{2}}, \quad \norm{u_{x}}_{\infty} \lesssim \norm{m}_{L^{2}}, \quad \norm{\rho}_{\infty} \lesssim \norm{\rho}_{H^{1}},
\end{equation*}
together with
\begin{equation*}
  \norm{\rho_{x}^{2}}_{L^{2}} \lesssim \norm{\rho}_{H^{1}}\norm{\rho}_{H^{2}},
\end{equation*}
we obtain
\begin{equation*}
  \int m_{x}m_{tx} \lesssim \left(1 + \norm{m}_{L^{2}} +
\norm{\rho}_{H^{1}}\right) \left(\norm{m}_{H^{1}}^{2} +
\norm{\rho}_{H^{2}}^{2}\right).
\end{equation*}
We have next
\begin{multline*}
  \int \rho_{xx}\rho_{txx} =  - \int u_{xx}\rho_{x}\rho_{xx} - 2 \int
u_{x}\rho_{xx}^{2} - \int u\rho_{xxx}\rho_{xx}
  \\
  - (a-1)\int u_{xxx}\rho\rho_{xx} - 2(a-1)\int u_{xx}\rho_{x}\rho_{xx}
- (a-1)\int u_{x}\rho_{xx}^{2},
\end{multline*}
from which we deduce
\begin{multline*}
  \int \rho_{xx}\rho_{txx} \lesssim
\norm{u}_{H^{2}}\norm{\rho_{x}^{2}}_{H^{1}} +
\norm{u_{x}}_{\infty}\norm{\rho}_{H^{2}}^{2} +
\norm{u_{x}}_{\infty}\norm{\rho}_{H^{2}}^{2}
  \\
  + \norm{\rho}_{\infty}\norm{u}_{H^{3}}\norm{\rho}_{H^{2}} +
\norm{u}_{H^{2}}\norm{\rho_{x}^{2}}_{H^{1}} +
\norm{u_{x}}_{\infty}\norm{\rho}_{H^{2}}^{2}.
\end{multline*}
Using the fact that $\HH{1}$ is a multiplicative algebra and the following estimates
\begin{equation*}
  \norm{u}_{H^{2}} \lesssim \norm{m}_{L^{2}}, \quad \norm{u_{x}}_{\infty} \lesssim \norm{m}_{L^{2}}, \quad \norm{u}_{H^{3}} \lesssim \norm{m}_{H^{1}}, \quad \norm{\rho}_{\infty} \lesssim \norm{\rho}_{H^{1}},
\end{equation*}
we thus obtain
\begin{equation*}
  \int \rho_{xx}\rho_{txx} \lesssim \left(\norm{m}_{L^{2}} +
\norm{\rho}_{H^{1}}\right) \left(\norm{m}_{H^{1}}^{2} +
\norm{\rho}_{H^{2}}^{2}\right),
\end{equation*}
which achieves the proof for $k=1$, using again Gronwall's lemma and the
first \textit{a priori} estimate obtained on $\norm{m}_{L^{2}} + \norm{\rho}_{H^{1}}$.
Assuming now that ~\eqref{eq:Gronwall-estimate} is true for $k \ge 1$,
we will show that the same is true for $k+1$. We have
\begin{multline*}
  \frac{d}{dt}\left(\norm{m}_{H^{k+1}}^{2} +
\norm{\rho}_{H^{k+2}}^{2}\right) =
\frac{d}{dt}\left(\norm{m}_{H^{k}}^{2} + \norm{\rho}_{H^{k+1}}^{2}\right)
  \\
  + 2 \int m_{t}^{(k+1)}m^{(k+1)} + 2 \int \rho_{t}^{(k+2)}\rho^{(k+2)}.
\end{multline*}
Thus, we need to estimate
\begin{multline}\label{eq:k+1-estimate}
  \alpha \int (u_{x})^{(k+1)}m^{(k+1)} - a\int (u_{x}m)^{(k+1)}m^{(k+1)}
- \int (um_{x})m^{(k+1)}
  \\
   - \kappa \int (\rho\rho_{x})^{(k+1)}m^{(k+1)} - \int
(u\rho_{x})^{(k+2)}\rho^{(k+2)} - (a-1)\int (u_{x}\rho)^{(k+2)}\rho^{(k+2)}.
\end{multline}
Using the following estimate derived from the Leibnitz formula:
\begin{multline*}
  \abs{\int (fg)^{(n+1)}h^{(n+1)}} \lesssim
\norm{f}_{H^{n}}\norm{g}_{H^{n+1}}\norm{h}_{H^{n+1}}
  \\
  + \abs{\int fg^{(n+1)}h^{(n+1)}} + \abs{\int f^{(n+1)}gh^{(n+1)}},
\end{multline*}
where $f,g,h \in \CS$ and $n \ge 1$, we deduce, using some integration
by parts, that each term in~\eqref{eq:k+1-estimate} is bounded, up to a
multiplicative constant by
\begin{equation*}
  \left(1 + \norm{m}_{H^{k}} + \norm{\rho}_{H^{k+1}}\right)
\left(\norm{m}_{H^{k+1}}^{2} + \norm{\rho}_{H^{k+2}}^{2}\right),
\end{equation*}
which achieves the proof.
\end{proof}

In~\cite{EK2014a} the following distance on $\D{q}$ was introduced,
\begin{equation*}
  d_{q}(\varphi_{1},\varphi_{2}) := d_{C^{0}}(\varphi_{1},\varphi_{2}) + \norm{\varphi_{1x} - \varphi_{2x}}_{H^{q-1}} + \norm{1/\varphi_{1x} - 1/\varphi_{2x}}_{\infty},
\end{equation*}
which makes $(\D{q},d_{q})$ a complete metric space. Following~\cite[Section~5]{EK2014a}, we can check that the spray $F_{q}$ defined in Theorem~\ref{thm:smoothness-spray} is bounded on bounded sets of $TH^{q}G$ (since the proof is similar, we redirect to this reference for the details). As a consequence of the \textit{a priori} estimates obtained in Corollary~\ref{cor:rho-bound} and Lemma~\ref{lem:Hk-apriori-estimate}, we obtain therefore the following theorem.

\begin{theorem}\label{thm:global-existence}
Let $(u_{0},\rho_{0}) \in \CS\times\CS$ and
\begin{equation*}
    (u,\rho) \in C^{\infty}(J,\CS\oplus\CS),
\end{equation*}
be the solution of~\eqref{eq:EHKL} with $(u(0), \rho(0)) = (u_{0},\rho_{0})$. Suppose that $\norm{u_{x}(t)}_{\infty}$ is bounded on any bounded
subinterval of $J$. Then, the solution $(u(t),\rho(t))$ is defined for $t \in \RR$.
\end{theorem}

\begin{remark}
Note that theorem~\ref{thm:global-existence} is still true for a solution
\begin{equation*}
    (u,\rho) \in C^{0}(J_{q},\HH{q}\oplus\HH{q-1}) \cap C^{1}(J_{q},\HH{q-1}\oplus\HH{q-2}),
\end{equation*}
for $q$ big enough but the proof requires the use of a Friedrich's mollifier and Kato-Ponce estimate (see~\cite[Theorem 5.1]{EK2014a}), to establish the \textit{a priori} estimate given in lemma~\ref{lem:Hk-apriori-estimate}.
\end{remark}

\begin{remark}
It is well-known that the Camassa-Holm equation (CH) has solutions which blow-up in finite time, cf. \cite{CE1998b}. One way to exhibit such solutions uses the fact that a solution which is spatially odd initially will remain spatially odd at each time. This is a consequence of the following symmetry of the equation
\begin{equation*}
  u(x) \mapsto -u(-x).
\end{equation*}
It could be interesting to note that equation~\eqref{eq:EHKL} is invariant under the transformation
\begin{equation*}
  u(x) \mapsto -u(-x), \quad \rho(x) \mapsto -\rho(-x), \quad \alpha \mapsto -\alpha,
\end{equation*}
but unfortunately not by the transformation
\begin{equation*}
  u(x) \mapsto -u(-x), \quad \rho(x) \mapsto -\rho(-x),
\end{equation*}
unless $\alpha = 0$. In particular, there is no reason, contrary to (CH), that a solution $(u,\rho)$ of~\eqref{eq:EHKL}, which is spatially odd for $t=0$ will remain odd at time $t$. There are classes other than odd initial data which lead to a finite time blow-up for the Camassa-Holm equation. The specification of compatible initial conditions for system \eqref{eq:Main} with $\alpha\ne 0$ remains open.
\end{remark}


\section*{Acknowledgements}

D. Henry, B. Kolev and T. Lyons were supported by the Irish Research Council--Campus France PHC ``Ulysses'' programme. All the authors would like to thank Rossen Ivanov for his stimulating discussions, and Cathy and Michel for their kind hospitality at \emph{Les Grandes Moli\`{e}res} during the preparation of this work. The authors would like to thank the anonymous referee for helpful suggestions and comments.


\end{document}